\documentclass[reqno,12pt]{amsart}

\usepackage[a4paper,
            left=1in,
            right=1in,
            top=1in,
            bottom=1in]{geometry}

\usepackage[T1]{fontenc}
\usepackage[utf8]{inputenc}
\usepackage{lmodern}
\usepackage{microtype}
\usepackage{amsmath,amssymb,amsthm,mathtools}
\usepackage{xcolor}
\usepackage{enumitem}
\usepackage{hyperref}
\usepackage[nameinlink,capitalize,noabbrev]{cleveref}
\usepackage{thmtools}
\usepackage{tikz}
\usepackage{graphicx}
\usetikzlibrary{arrows.meta,positioning,fit,calc}

\usepackage{caption}
\definecolor{hornblue}{RGB}{58,82,112}
\definecolor{hornaccent}{RGB}{117,89,67}
\definecolor{hornfill}{RGB}{244,247,250}
\definecolor{hornwarmfill}{RGB}{249,246,242}
\definecolor{horngray}{RGB}{92,92,92}
\tikzset{
  horn flow/.style={
    draw=hornblue,
    fill=hornfill,
    rounded corners=2pt,
    line width=.65pt,
    align=center,
    inner xsep=6pt,
    inner ysep=5pt,
    font=\small
  },
  horn horn/.style={
    draw=hornaccent,
    fill=hornwarmfill,
    rounded corners=2pt,
    line width=.65pt,
    align=center,
    inner xsep=6pt,
    inner ysep=5pt,
    font=\small
  },
  horn component/.style={
    draw=hornblue,
    fill=hornfill,
    rounded corners=3pt,
    line width=.65pt,
    inner sep=7pt
  },
  horn element/.style={
    draw=horngray,
    fill=white,
    rounded corners=1.5pt,
    line width=.55pt,
    inner xsep=4pt,
    inner ysep=3pt,
    font=\scriptsize
  },
  horn point/.style={
    circle,
    draw=horngray,
    fill=white,
    line width=.55pt,
    minimum size=5.6mm,
    inner sep=0pt,
    font=\scriptsize
  },
  horn math/.style={
    draw=horngray,
    fill=white,
    rounded corners=2pt,
    line width=.55pt,
    align=center,
    inner xsep=5pt,
    inner ysep=4pt,
    font=\small
  },
  horn arrow/.style={-{Latex[length=2.2mm,width=1.5mm]}, line width=.65pt, draw=horngray},
  horn relation/.style={-{Latex[length=2mm,width=1.4mm]}, line width=.6pt, draw=hornblue},
  horn dashed/.style={-{Latex[length=2mm,width=1.4mm]}, line width=.6pt, dashed, draw=hornaccent},
  horn label/.style={font=\scriptsize, text=horngray, align=center}
}

\hypersetup{
  pdftitle={Exponential Lower Bounds for Horn Amalgamation},
  pdfauthor={Jakub Rydval},
  pdfsubject={Complexity lower bounds for amalgamation of finite models of universal Horn sentences},
  pdfkeywords={amalgamation, universal Horn sentences, EXPTIME, 2EXPTIME, tree automata, corridor games},
  colorlinks=true,
  linkcolor=blue!50!black,
  citecolor=blue!50!black,
  urlcolor=blue!50!black
}

\declaretheorem[name=Theorem,numberwithin=section]{theorem}

\declaretheorem[name=Lemma,sibling=theorem]{lemma}

\newcommand{\struct}[1]{\mathfrak{#1}}
\newcommand{\fm}{\operatorname{fm}}

\newcommand{\ar}{\operatorname{ar}}
\newcommand{\ComplexityClass}[1]{\mbox{\normalfont\textsc{#1}}}

\title[Exponential Lower Bounds for Horn Amalgamation]{Exponential Lower Bounds for Horn Amalgamation}
\author{Jakub Rydval}
\thanks{This research was funded in whole or in part by the Austrian Science Fund (FWF) [ESP 1571225]. For the purpose of Open Access, the author has applied a CC BY public copyright licence to any Author Accepted Manuscript (AAM) version arising from this submission.}

\begin{document}
\begin{abstract}
The amalgamation decision problem asks whether the finite models of a given
universal first-order sentence have the amalgamation property. We prove that
this problem is \ComplexityClass{ExpTime}-hard for universal Horn
sentences of maximum relation arity three, and
\ComplexityClass{2ExpTime}-hard when the arity is unbounded, thereby
strengthening the previous best lower bounds for universal
Horn sentences.
              
The first reduction simulates runs of nondeterministic tree automata using a single ternary relation and unary and binary auxiliary relations. 
The second
simulates finite winning strategies in an exponential corridor game. 
As a secondary contribution, we
prove that deciding the winner remains \ComplexityClass{2ExpTime}-complete
for the game normal form used in our reduction, in which each move label determines a
unique legal successor from every non-winning legal row.
\end{abstract}
\maketitle

\section{Introduction}\label{sec:introduction}

A class of finite structures over a finite relational signature has the amalgamation property if any two of
its members that agree on a common induced substructure can be embedded
into a third member so that the embeddings agree on the common part.
By Fra\"iss\'e's theorem, a nonempty class of finite structures closed under
isomorphisms and induced substructures is the age of a countable homogeneous
structure if and only if it has the amalgamation property. The theorem therefore turns the existence of a potentially infinite
homogeneous structure into a question about finite amalgamation diagrams.

For a universal first-order sentence $\Phi$, write $\fm(\Phi)$ for the
class of its finite models. The \emph{amalgamation decision problem} (ADP)
asks, on input $\Phi$, whether $\fm(\Phi)$ has the amalgamation property
(AP). Universal sentences give finite descriptions of hereditary classes:
every failure of a sentence is already witnessed on the elements assigned
to its variables. This makes the ADP the computational counterpart
of the study of \emph{finitely bounded} homogeneous structures.

In the present article, we study the ADP for \emph{universal Horn sentences}, given
explicitly as conjunctions of universally quantified Horn clauses. Such a
clause has a conjunction of atoms as its premise and either one atom or
falsity as its conclusion. The syntactic restriction is substantial, but
it still permits computations to be expressed by repeated application of
Horn rules. The question is how much computational complexity can be
encoded in the amalgamation property of the resulting finite model class.

Earlier work established \ComplexityClass{PSpace}-hardness for Horn
amalgamation at arity three using word computations
\cite[Theorem~5]{RydvalICALP2024}. A related corridor construction established
\ComplexityClass{ExpSpace}-hardness for succinct universal presentations of
Horn theories: its conversion into explicitly listed Horn clauses incurs
an exponential increase in size
\cite[Theorem~6 and its proof, p.~150:12]{RydvalICALP2024}.
For unrestricted universal input, the
stronger lower bound of \ComplexityClass{2NExpTime}-hardness is known
from the expanded journal version~\cite[Corollary~1.5]{RydvalInsideOut}
of~\cite{RydvalICALP2024}; for input given as a set of forbidden
substructures, the corresponding result is \ComplexityClass{NExpTime}-hardness
\cite[Corollary~1.6]{RydvalInsideOut}. The contribution here concerns the
restriction to explicitly given Horn clauses. We use branching computations
to obtain stronger lower bounds for this restriction.
Universality of nondeterministic finite tree automata yields the first result.

\begin{theorem}\label{thm:bounded-arity-hardness}
The ADP for universal Horn sentences of relation arity
at most three is \ComplexityClass{ExpTime}-hard. More precisely, it is
\ComplexityClass{ExpTime}-hard to distinguish sentences whose finite models
have the SAP from sentences whose finite models do not have the AP.
This holds for equality-free sentences with only one ternary relation
symbol and all remaining relation symbols unary or binary.
\end{theorem}

Here the SAP denotes the \emph{strong amalgamation property}: the images of
the two structures in the amalgam intersect exactly in the image of their
common part. The stronger distinction in the theorem gives the same
lower bound for deciding the SAP. In particular, the bound applies to every
fixed maximum relation arity at least three.

For unbounded arity, we use a game played on rows of exponential width.
The moves simulate an alternating exponential-space computation, and a
winning strategy is represented by a finite tree of rows.

\begin{theorem}\label{thm:unbounded-arity-hardness}
The ADP for universal Horn sentences is
\ComplexityClass{2ExpTime}-hard when relation arities are unbounded.
More precisely, it is \ComplexityClass{2ExpTime}-hard to distinguish
sentences whose finite models have the SAP from sentences whose finite
models do not have the AP. The lower bound holds for equality-free
sentences.
\end{theorem}

The corridor game requires a specific normal form. There is a unique
legal initial row; every non-winning legal row is owned by one of two
players; and, from every such row, each of two move labels determines
exactly one legal successor row.
We prove that deciding the winner remains
\ComplexityClass{2ExpTime}-complete in this normal form. Uniqueness of
successors is essential to the Horn encoding: at a universally owned row,
one Horn rule checks a winning successor for each label, and these two
successors must account for all available moves.

Both reductions use the same amalgamation mechanism. Apart from designated
seed clauses, every clause is \emph{complete}: each pair of distinct
variables occurs together in a premise atom. Such clauses are preserved
under free amalgams over common induced substructures. Each seed clause
omits just one pair, consisting of two distinguished variables called
\emph{poles}. Placing the poles on opposite sides of an amalgamation
diagram activates a computation over its common part. Padding atoms ensure
that every tuple newly derived by Horn closure still has its poles on
opposite sides and all its computational data in the intersection.

In the tree-automaton reduction, an unaccepted tree gives a diagram with no
amalgam, whereas universality allows every diagram to be completed to a
strong amalgam. In the corridor-game reduction, a winning strategy gives
a diagram with no amalgam; if there is no winning strategy, Horn closure
again constructs a strong amalgam for every diagram. These arguments
work directly with the input structures and their unions.

The results establish lower bounds for the Horn syntax. The
decidability of the ADP for universal Horn sentences in arbitrary arity remains an open problem,
posed in \cite[Question~1]{schrottenloher2022universal}.
For context, the ADP over binary
signatures is decidable even for unrestricted universal sentences
\cite[Theorem~4]{BodirskyKnauerStarkeASNP}, with a
\ComplexityClass{coNExpTime} upper bound for universal input
\cite[Proposition~3.2]{RydvalInsideOut}. 
Over binary signatures, explicitly given universal Horn input admits the
better upper bound $\Pi^{p}_2$~\cite[Section~6.3, Proposition~22,
p.~79, and the preceding remark]{rydval2022CSP}.
The ternary relation in
\cref{thm:bounded-arity-hardness} is what lets the first construction carry
two poles together with a tree position.

\Cref{sec:preliminaries} introduces the required notation and the
free-union and Horn-closure tools. \Cref{sec:tree-automata} proves
\cref{thm:bounded-arity-hardness}. \Cref{sec:corridor-games} defines the
exponential corridor game and proves completeness of its binary normal
form. \Cref{sec:unbounded-arity} translates these games into Horn
sentences and proves \cref{thm:unbounded-arity-hardness}.
\Cref{sec:conclusion} discusses the remaining questions.

\section{Preliminaries}\label{sec:preliminaries}

We write $\mathbb N=\{0,1,2,\ldots\}$, $[n]=\{1,\ldots,n\}$, and
$\bar x$ for a tuple of variables or elements. All reductions are
polynomial-time many-one reductions. Input sentences are represented by
explicit lists of clauses; their length includes all occurrences of
relation symbols and variables. In particular, arities and the lengths
of tuples are charged to the input size.

\subsection{Relational structures and amalgamation}

A \emph{relational signature} $\tau$ is a finite set of relation symbols,
each with a specified positive arity. A \emph{$\tau$-structure}
$\struct A$ consists of a domain $A$ and a relation
$R^{\struct A}\subseteq A^k$ for each $k$-ary $R\in\tau$.
Empty domains are allowed. The \emph{induced substructure} on
$B\subseteq A$ has relations $R^{\struct A}\cap B^k$.

A map $h\colon A\to B$ is a \emph{homomorphism} from $\struct A$ to
$\struct B$ if it preserves every relation. It is an \emph{embedding} if
it is injective and also reflects every relation: for each $R\in\tau$
and each tuple $\bar a$ of the appropriate length, we have 
 $\bar a\in R^{\struct A}$ if and only if 
 $h(\bar a)\in R^{\struct B}.$ 
Here maps act on tuples coordinatewise. A surjective embedding is an
\emph{isomorphism}.

Let $\mathcal K$ be a class of finite $\tau$-structures closed under
isomorphisms and induced substructures. An \emph{amalgamation diagram}
in $\mathcal K$ consists of $\struct B_1,\struct B_2\in\mathcal K$
inducing the same structure $\struct B_0$ on $B_0=B_1\cap B_2$.
An \emph{amalgam} is a structure $\struct C\in\mathcal K$ with
embeddings $f_i\colon\struct B_i\to\struct C$ for $i\in\{1,2\}$
such that $f_1|_{B_0}=f_2|_{B_0}$. The class has the \emph{amalgamation
property} (AP) if every such diagram has an amalgam. It has the
\emph{strong amalgamation property} (SAP) if the amalgam can always be
chosen with
$ 
 f_1(B_1)\cap f_2(B_2)=f_1(B_0).
$ 

The \emph{free amalgam} $\struct B_1\cup\struct B_2$ of an amalgamation
diagram has domain $B_1\cup B_2$ and relations
$R^{\struct B_1}\cup R^{\struct B_2}$. The class has the \emph{free
amalgamation property} (FAP) if it is closed under taking free amalgams.
Thus FAP implies SAP, and SAP implies AP. Since our classes are
hereditary, any strong amalgam can be restricted to the union of the
two images. After renaming elements, it then has domain $B_1\cup B_2$
and is obtained by adding relation tuples to the free amalgam without
changing either induced side.

Diagrams may have empty common part. Consequently the AP implies the joint
embedding property for a nonempty hereditary class. A structure is
\emph{homogeneous} if every isomorphism between its finite induced
substructures extends to an automorphism. Its \emph{age} is the class
of finite structures that embed into it. Fra\"iss\'e's theorem says that
a nonempty class closed under isomorphisms and induced substructures, in a
finite relational signature, is the age of a countable homogeneous
structure if and only if it has the AP; that
structure is unique up to isomorphism.

\subsection{Universal Horn sentences}

A \emph{universal sentence} has the form
$\forall\bar x\,\phi(\bar x)$ with $\phi$ quantifier-free. A
\emph{universal Horn sentence}, or simply a \emph{Horn sentence},
is given as a finite conjunction of universally quantified clauses
$ 
 \alpha_1\wedge\cdots\wedge\alpha_m\Rightarrow\beta,
$ 
where each $\alpha_j$ is a relational or equality atom and $\beta$ is
such an atom or $\bot$ (falsity). The premise may be empty. We suppress
universal quantifiers when displaying clauses. A clause with a relational
head is called a \emph{positive-head clause}; a clause with an equality
head is called an \emph{equality-head clause}; and a clause with head $\bot$
is called a \emph{constraint}. A sentence is \emph{equality-free} if it
contains no equality atoms. All sentences constructed in the reductions
are equality-free.

We write $\fm(\Phi)$ for the finite models of a universal sentence
$\Phi$. This class is closed under isomorphisms and induced
substructures. The \emph{Horn ADP} takes an explicitly given Horn
sentence $\Phi$ as input and asks whether $\fm(\Phi)$ has the AP.
The corresponding problem for SAP is defined in the same way.

\subsection{Complete clauses and Horn closure}

The \emph{premise Gaifman graph} of a Horn clause has all variables
occurring in the clause as vertices. Two distinct variables are adjacent
if they occur together in a premise atom. The clause is \emph{complete}
if this graph is complete. Variables occurring only in the head still
count as vertices. For example,
$ 
 E(x,y)\wedge E(x,z)\wedge E(y,z)\Rightarrow\bot
$ 
is complete, whereas $E(x,y)\wedge E(x,z)\Rightarrow E(y,z)$ is not.
A Horn sentence is complete if all of its clauses are complete.

\begin{lemma}\label{lem:complete-free-union}
Let $\struct B_1$ and $\struct B_2$ satisfy a complete Horn clause and
induce the same substructure on $B_1\cap B_2$. Their free amalgam also
satisfies the clause. Consequently, the finite models of a complete
Horn sentence have the FAP.  
\end{lemma}

\begin{proof}
Suppose that an assignment to all variables of the clause makes the premise true
in $\struct{B}_1\cup \struct{B}_2$. Recall that the premise Gaifman graph also includes variables
occurring only in the head. The image of the assignment
cannot meet both $B_1\setminus B_2$ and $B_2\setminus B_1$:
completeness would put values from opposite sides together in a premise
atom. A relational tuple containing these values belongs to neither
side, and an equality between them is false. Thus the assignment
lies in one side  $B_i$. Agreement on the common induced
substructure implies that every premise atom also holds in
$\struct B_i$. The clause holds there, so its head holds in the union;
if the head is $\bot$, the premise cannot have held in the first place.
\end{proof}

For a finite structure $\struct U$ and an equality-free Horn sentence, the
\emph{Horn closure} under the positive-head clauses is defined inductively. Set $\struct U^{(0)}=\struct U$. Obtain
$\struct U^{(r+1)}$ by adding the head tuple of every positive-head
clause under every assignment whose premise holds in $\struct U^{(r)}$.
The domain does not change. There are only finitely many possible
relation tuples, so this increasing sequence reaches a fixed point.
The fixed point satisfies all positive-head clauses. It is a model of
the sentence precisely when it also satisfies the constraints.

We use the same construction for selected subsets of the positive-head
clauses. Under any homomorphism from the initial structure into a model of
those clauses, the image of every derived tuple lies in the corresponding
relation, by induction on its derivation stage. In
particular, a derivation of $\bot$ from the facts of an amalgamation
diagram rules out every amalgam, including those identifying elements
from opposite sides. Conversely, a closure that satisfies the sentence
and changes no induced side is a strong amalgam.

In both reductions, the only incomplete clauses are seed clauses whose
premise Gaifman graph omits just the pair of pole variables $(y_1,y_2)$.
All other pairs are covered by premise atoms, including suitable
\emph{padding atoms}. The extraction lemmas below
(\cref{lem:cross-q-tree-extraction,lem:corridor-cross-extraction})
show that every new
tuple has its poles in opposite sides outside the intersection and all
remaining entries in the intersection. This is the invariant that
ensures Horn closure preserves both induced sides.

\section{The lower bound at arity three}\label{sec:tree-automata}

\subsection{Tree-automaton universality}\label{sec:tree-universality}

For bounded arity, we replace the word-based computation in the
$\ComplexityClass{PSpace}$-hardness proof of
\cite[Theorem~5]{RydvalICALP2024} by a tree-based
computation, reducing from universality of regular tree languages.
We use
the standard $\ComplexityClass{ExpTime}$-complete universality problem for
nondeterministic bottom-up tree automata~\cite{SeidlTreeAutomata}; see also
\cite[Section~1.7, Theorem~13]{TATA} for an explicit statement of universality hardness.

Let $\mathbb{N}^*$ denote the set of all finite words over $\mathbb{N}$.
A \emph{ranked alphabet} is a finite set \(\Sigma\) equipped with an arity
map
$
  \ar\colon\Sigma\to\mathbb N.
$
A finite \emph{\(\Sigma\)-tree} is a map \(t\colon D\to\Sigma\), where
\(D\subseteq\mathbb N^*\) is finite, nonempty, and prefix-closed, such that,
for every \(u\in D\) and \(i\in\mathbb N\), we have
$
  ui\in D$ if and only if
  $0\leq i<\ar(t(u))$.
Thus a node labelled by \(f\in\Sigma\) has exactly \(\ar(f)\) ordered
children. The \emph{root} is the empty word \(\varepsilon \in \mathbb{N}^{*}\). We write
\(T_\Sigma\) for the set of all finite \(\Sigma\)-trees.

A \emph{nondeterministic bottom-up tree automaton} (NBTA) is a tuple
$
  \mathcal A=(S,\Sigma,\Delta,F),
$
where \(S\) is a finite set of states, \(F\subseteq S\) is the set of
accepting states, and \(\Delta\) is a finite set of transitions of the form
$
  f(s_0,\ldots,s_{k-1})\rightarrow s$ for
  $f\in\Sigma, $  $\ar(f)=k,$ and  $s_0,\ldots,s_{k-1},s\in S.
$
For \(k=0\), such a transition is written \(f\to s\).

A \emph{run} of \(\mathcal A\) on a tree \(t\colon D\to\Sigma\) is a map
\(\rho\colon D\to S\) such that, for every \(u\in D\), writing
\(f=t(u)\) and \(k=\ar(f)\), the set of transitions $\Delta$ contains
\[
  f\bigl(\rho(u0),\ldots,\rho(u(k-1))\bigr)
  \rightarrow \rho(u).
\]
The run is accepting if \(\rho(\varepsilon)\in F\). The language recognized
by \(\mathcal A\) is
\[
  L(\mathcal A)
  =
  \{t\in T_\Sigma \mid
    \mathcal A\text{ has an accepting run on }t\}.
\]
A tree language \(L\subseteq T_\Sigma\) is called
\emph{regular} if it is recognized by an NBTA,
that is, if \(L=L(\mathcal A)\) for some such automaton \(\mathcal A\).
The automaton \(\mathcal A\) is \emph{universal} if
$
  L(\mathcal A)=T_\Sigma.
$

Universality remains \(\ComplexityClass{ExpTime}\)-complete when the ranked
alphabet has maximum rank two. We give the rank reduction explicitly.
If the alphabet has no nullary symbol, there are no finite trees and the
automaton is universal. If there is a nullary symbol but some symbol has no
transition, the automaton is not universal. These cases can be mapped to
fixed universal and non-universal binary automata, respectively. Otherwise,
every symbol occurs in a transition, so the sum of the alphabet's ranks is
bounded by the size of the explicitly listed transition table.

Keep symbols of rank at most two. For each symbol $f$ of rank $k>2$,
introduce distinct binary symbols $f_1,\ldots,f_{k-1}$, and encode a node by
\[
 \operatorname{enc}\bigl(f(t_0,\ldots,t_{k-1})\bigr)
 =f_1\bigl(\operatorname{enc}(t_0),
       f_2\bigl(\operatorname{enc}(t_1),\ldots,
       f_{k-1}(\operatorname{enc}(t_{k-2}),
               \operatorname{enc}(t_{k-1}))\ldots\bigr)\bigr).
\]
The map $\operatorname{enc}$ acts recursively and leaves symbols of rank at
most two unchanged. The indexed symbols make decoding unique, so it is a
bijection between $T_\Sigma$ and the well-formed encodings. A complete
deterministic bottom-up automaton recognizes the latter using one state for
a complete encoded subtree, one for each permitted proper suffix of an
$f$-chain, and an error state. Its transitions check the next chain index
and require every child attached to the chain to be a complete encoded
subtree. Both its state set and its transition table have polynomial size;
complementing its accepting states recognizes all malformed encodings.

To simulate an original transition
$\delta=(f(s_0,\ldots,s_{k-1})\to s)$, an NBTA uses states
$(\delta,j)$ along the chain to check the required child states successively,
and outputs $s$ at its root. There is one auxiliary state for each relevant
position of each explicitly listed transition. Taking the union of this
simulator with the automaton accepting malformed encodings gives a
polynomial-size binary automaton that is universal precisely when the
original automaton is universal.

\subsection{The Horn encoding}

For the rest of this section, let
$\mathcal A=(S,\Sigma,\Delta,F)$ be an NBTA over a ranked alphabet of
rank at most two.
We construct an equality-free universal Horn sentence
\(\Phi_{\mathcal A}\).  Its signature contains the unary symbols
  $L,\ R,\ I,\ T$, $A_f$ $(f\in\Sigma)$, the
binary symbols
$
D,C_0,C_1$, $P_s$ $(s\in S)$,
and one ternary
symbol \(Q\).  The symbols \(A_f,C_0,C_1\) describe labelled ordered trees;
\(D\) is a padding relation; \(P_s(y,x)\) says that the subtree rooted at
\(x\) admits a run with root state \(s\), anchored at the pole  \(y\); and
\(Q(y_1,y_2,x)\) is the two-pole predicate used to parse an arbitrary tree.
The auxiliary marker \(I\) selects the anchors subject to the acceptance
constraints. We retain it to keep this role separate from the pole markers
\(L,R\).

We write \(\Phi_{\mathcal A}=\Phi_{\mathrm{acc}}\wedge
\Phi_{\mathrm{all}}\).  For every transition
\(\delta=(f(s_0,\ldots,s_{k-1})\to s)\), the first part $\Phi_{\mathrm{acc}}$ contains
\begin{equation}
\begin{split}
 L(y)\wedge I(y)\wedge A_f(x)\wedge D(y,x)
 {}\wedge\bigwedge\nolimits_{i<k}
   \bigl(P_{s_i}(y,x_i)\wedge C_i(x,x_i)\bigr)\\
 {}\wedge
 \begin{cases}
   D(x_0,x_1),&k=2,\\
   \top,&k<2
 \end{cases}
 \quad\Rightarrow\quad P_s(y,x).
\end{split}
\label{eq:tree-acc-rule}
\end{equation}
For every \(s\in F\), it also contains
\begin{equation}
  I(y)\wedge T(x)\wedge P_s(y,x)\Rightarrow\bot.
  \label{eq:tree-acc-final}
\end{equation}

The second part $\Phi_{\mathrm{all}}$ contains, for every nullary \(a\in\Sigma\), the \emph{seed clause}
\begin{equation}
\begin{split}
 L(y_1)\wedge R(y_2)\wedge I(y_1)\wedge A_a(x)
 {}\wedge D(y_1,x)\wedge D(y_2,x)
 \Rightarrow Q(y_1,y_2,x).
\end{split}
\label{eq:tree-parse-seed}
\end{equation}
For every \(f\in\Sigma\) of positive rank \(k\), it contains the \emph{recursive
clause}
\begin{equation}
\begin{split}
 A_f(x)\wedge D(y_1,x)\wedge D(y_2,x)
 {}\wedge\bigwedge\nolimits_{i<k}
   \bigl(Q(y_1,y_2,x_i)\wedge C_i(x,x_i)\bigr)\\
 {}\wedge
 \begin{cases}
   D(x_0,x_1),&k=2,\\
   \top,&k=1
 \end{cases}
 \quad\Rightarrow\quad Q(y_1,y_2,x),
\end{split}
\label{eq:tree-parse-rule}
\end{equation}
and finally
\begin{equation}
  Q(y_1,y_2,x)\wedge T(x)\Rightarrow\bot.
  \label{eq:tree-parse-final}
\end{equation}
Here and in \eqref{eq:tree-acc-rule} an occurrence of \(\top\) denotes the
empty conjunction.
The construction is polynomial in \(\mathcal A\), and \(Q\) is its only
ternary relation symbol.

\begin{lemma}\label{lem:tree-acc-complete}
Every clause of \(\Phi_{\mathrm{acc}}\) is complete.  Every non-seed clause
of \(\Phi_{\mathrm{all}}\), including
\eqref{eq:tree-parse-final}, is complete.  The only non-complete clauses are
the seed clauses \eqref{eq:tree-parse-seed}, whose premise Gaifman graph is
missing precisely the edge \((y_1,y_2)\).
\end{lemma}

\begin{proof}
For \(k=2\), the six pairs among \(y,x,x_0,x_1\) in
\eqref{eq:tree-acc-rule} occur in the atoms
$
  D(y,x)$, $P_{s_0}(y,x_0)$, $P_{s_1}(y,x_1)$,
  $C_0(x,x_0)$, $C_1(x,x_1)$,  $D(x_0,x_1)$.
The cases \(k=0,1\) follow by deleting the unused child variables, and
\eqref{eq:tree-acc-final} is complete because \(P_s(y,x)\) contains both
variables. In \eqref{eq:tree-parse-rule}, each atom
\(Q(y_1,y_2,x_i)\) supplies the three edges on
\(\{y_1,y_2,x_i\}\).  The atoms \(D(y_j,x)\), \(C_i(x,x_i)\), and, in the
binary case, \(D(x_0,x_1)\), supply all remaining pairs.  The assertions about
\eqref{eq:tree-parse-final} and \eqref{eq:tree-parse-seed} are immediate.
\end{proof}

We next make the two computations performed by the Horn clauses explicit.
For \(t\in T_\Sigma\), let \(\operatorname{pos}(t)\) be its set of positions,
let \(f_u\) be the label at \(u\), and choose a variable \(x_u\) for each
position.  Write \(t|_u\) for the subtree rooted at \(u\).  A state \(s\) is
\emph{reachable at \(u\)} if \(\mathcal A\) has a run on \(t|_u\) whose root
state is \(s\).  Set
\begin{align}
 \theta_t(\bar x)
  \coloneqq{}&
  \bigwedge\nolimits_{u\in\operatorname{pos}(t)}A_{f_u}(x_u)
  \wedge
  \bigwedge\nolimits_{\substack{u\in\operatorname{pos}(t)\\i<\ar(f_u)}}
        C_i(x_u,x_{ui})
  \wedge
  \bigwedge\nolimits_{\substack{u\in\operatorname{pos}(t)\\\ar(f_u)=2}}
        D(x_{u0},x_{u1}),
 \label{eq:tree-pattern}\\
 \lambda_t(y,\bar x)
  \coloneqq{}&
  L(y)\wedge I(y)\wedge\theta_t(\bar x)
  \wedge\bigwedge\nolimits_{u\in\operatorname{pos}(t)}D(y,x_u),
 \label{eq:one-pole-tree-pattern}\\
 \eta_t(y_1,y_2,\bar x)
  \coloneqq{}&
  L(y_1)\wedge R(y_2)\wedge I(y_1)\wedge T(x_\varepsilon)
  \wedge\theta_t(\bar x)\notag\\
 &{}\wedge\bigwedge\nolimits_{u\in\operatorname{pos}(t)}
       \bigl(D(y_1,x_u)\wedge D(y_2,x_u)\bigr).
 \label{eq:two-pole-tree-pattern}
\end{align}

\begin{lemma}\label{lem:tree-horn-simulation}
For every \(t\in T_\Sigma\) and \(s\in S\),
\[
 \Phi_{\mathrm{acc}}\models
 \forall y,\bar x\,
   \bigl(\lambda_t(y,\bar x)  \Rightarrow  P_s(y,x_\varepsilon)\bigr)
\]
if and only if \(\mathcal A\) has a run on \(t\) whose root state is \(s\).
Moreover,
\[
 \Phi_{\mathrm{all}}\models
 \forall y_1,y_2,\bar x\,
   \bigl(\eta_t(y_1,y_2,\bar x) \Rightarrow \bot\bigr).
\]
\end{lemma}

\begin{proof}
For the first statement, suppose that \(\mathcal A\) has a run on \(t\) with
root state \(s\).  Induction on \(t\), using the transition selected by the
run at each node in \eqref{eq:tree-acc-rule}, derives
\(P_s(y,x_\varepsilon)\).  Conversely, if there is no such run, take the
structure \(\struct D_t\) with domain
\(\{y\}\cup\{x_u\mid u\in\operatorname{pos}(t)\}\) whose relational facts
are precisely the instances displayed in \(\lambda_t\), interpret all other
relations, including \(T\) and every \(P_s\), as empty, and close it under the
positive-head clauses of \(\Phi_{\mathrm{acc}}\).  Since the position
variables are distinct and its only \(A_f\)- and \(C_i\)-facts are those of
\(t\), induction shows that the resulting least model contains
\(P_s(y,x_u)\) exactly for the states reachable at position \(u\).  It is
therefore a countermodel to the displayed implication whenever \(s\) is not
reachable at the root.

For the second statement, \eqref{eq:tree-parse-seed} derives \(Q\) at every
leaf.  Successive applications of \eqref{eq:tree-parse-rule} derive \(Q\) at
each internal node, and \eqref{eq:tree-parse-final} derives \(\bot\) at the
root.
\end{proof}

Only the run-to-entailment direction of the first equivalence is needed
below; the converse records that the Horn simulation derives no spurious root
states.

\subsection{Extracting trees and constructing amalgams}

The following lemma is the tree-shaped replacement for the path-extraction
argument in the regular-grammar reduction of~\cite{RydvalICALP2024}.  Notice that it concerns
only newly generated \(Q\)-tuples and therefore does not require the input
structures themselves to be trees.

\begin{lemma}\label{lem:cross-q-tree-extraction}
Let \(\struct B_1,\struct B_2\) be an amalgamation diagram in
\(\fm(\Phi_{\mathcal A})\), with common part \(\struct B_0\) and free
amalgam \(\struct U\). Close \(\struct U\) under the \(Q\)-rules
\eqref{eq:tree-parse-seed} and \eqref{eq:tree-parse-rule}.
Every newly added tuple \(Q(a,b,c)\) satisfies, after exchanging the two
sides if necessary,
\[
 a\in B_1\setminus B_0,\qquad b\in B_2\setminus B_0,
\]
and admits a tree \(t\in T_\Sigma\) and a map
\(h\colon\operatorname{pos}(t)\to B_0\), with \(h(\varepsilon)=c\), such that
\[
 \struct U\models
 \lambda_t(a,\bar h)\wedge R(b)\wedge
 \bigwedge_{u\in\operatorname{pos}(t)}D(b,h(u)),
\]
where \(\bar h=(h(u))_{u\in\operatorname{pos}(t)}\).
\end{lemma}

\begin{proof}
We proceed by induction on the stage at which \(Q(a,b,c)\) is first added.  Suppose first
that it is added by a seed clause.  Neither pole can lie in \(B_0\).  Indeed,
if, say, \(a\in B_0\) and \(b\notin B_0\), let \(\struct B_i\) be the side
containing \(b\).  The atom \(D(b,c)\) forces \(c\in B_i\).  If both poles
lie in \(B_0\), instead choose a side containing \(c\).  In either case every
premise atom belongs to the chosen side; for tuples contained in \(B_0\),
this uses that the two sides induce the same structure there.  Since that side
satisfies the seed clause, the head would already belong to
\(Q^{\struct U}\), a contradiction.  The case in which only \(b\) lies in
\(B_0\) is symmetric.  If the two poles lay outside \(B_0\) on the same side,
the atoms \(D(a,c)\) and \(D(b,c)\) would again put the entire premise in
that side, giving the same contradiction.  Thus \(a,b\) lie in opposite sides
outside \(B_0\), and the two atoms \(D(a,c)\) and \(D(b,c)\) force
\(c\in B_0\).
The remaining assertions hold for the one-node tree labelled by the nullary
symbol used in the seed clause.

Suppose next that the tuple is added by a recursive clause.  At least one of
its premise \(Q\)-atoms lies outside \(Q^{\struct U}\).  Indeed, the other
premise relations are not changed during the closure, so if all premise
\(Q\)-atoms belonged to \(Q^{\struct U}\),
\cref{lem:complete-free-union} applied to the complete recursive clause
would put its head in \(Q^{\struct U}\). Such a premise tuple was
added at an earlier closure stage.  By the induction hypothesis, it has poles
\(a,b\) in opposite sides and its root lies in \(B_0\).  The atoms
\(D(a,c)\) and \(D(b,c)\) force the new root
\(c\) into \(B_0\).  Every other premise \(Q(a,b,d_i)\) is also a cross tuple
and therefore cannot belong to \(Q^{\struct U}\); by induction it supplies a
tree \(t_i\) rooted at \(d_i\in B_0\).
Attach these trees below a new root labelled by \(f\). The atom \(A_f(c)\),
the \(C_i\)-atoms and,
in the binary case, \(D(d_0,d_1)\), supply exactly the remaining conjuncts of
\(\theta_t\).  This also proves the assertions concerning both pole anchors.
\end{proof}

\begin{lemma}\label{lem:tree-automaton-ap-equivalence}
For every NBTA \(\mathcal A\) over a ranked alphabet of rank at most two, we have
\[
  L(\mathcal A)=T_\Sigma
  \quad\Longleftrightarrow\quad
  \fm(\Phi_{\mathcal A})\text{ has the SAP}
  \quad\Longleftrightarrow\quad
  \fm(\Phi_{\mathcal A})\text{ has the AP}.
\]
\end{lemma}

\begin{proof}
Suppose first that \(L(\mathcal A)\neq T_\Sigma\), and fix
\(t\in T_\Sigma\setminus L(\mathcal A)\). Take one distinct element
\(x_u\) for every position of \(t\), and let the relational facts of
\(\struct B_0\) be precisely the instances displayed in
\(\theta_t(\bar x)\), together with \(T(x_\varepsilon)\).
Extend it to \(\struct B_1\) by a new element \(a\), the facts
$
 L(a)$, $I(a)$, and $D(a,x_u)$ for all $u\in\operatorname{pos}(t)$,
and the facts \(P_s(a,x_u)\) for exactly those states \(s\) reachable at
position \(u\).  Extend \(\struct B_0\) to \(\struct B_2\) by a new element
\(b\), the fact \(R(b)\), and all facts \(D(b,x_u)\).  All unmentioned tuples
are absent.

The transition semantics of \(\mathcal A\) shows that
\(\struct B_1\models\Phi_{\mathrm{acc}}\); clause
\eqref{eq:tree-acc-final} does not fire because no final state is reachable at
the root. The sentence \(\Phi_{\mathrm{acc}}\) holds in \(\struct B_2\)
because \(L\) and all \(P_s\) are empty there.
No clause of \(\Phi_{\mathrm{all}}\) fires in either side:
\(Q\) is empty in both sides, disabling the recursive and final clauses,
while the seed clauses require \(R\), which is empty in \(\struct B_1\),
and \(L\), which is empty in \(\struct B_2\).
Hence \(\struct B_1,\struct B_2\) is an amalgamation diagram
in \(\fm(\Phi_{\mathcal A})\).
In any amalgam, the images of these elements satisfy the premise
\(\eta_t(a,b,\bar x)\).  By \cref{lem:tree-horn-simulation}, this is
impossible in a model of \(\Phi_{\mathrm{all}}\).  Thus the diagram has no
amalgam, and \(\fm(\Phi_{\mathcal A})\) does not have the AP.
The diagram and the obstruction forced in every prospective amalgam are
illustrated in \cref{fig:tree-automaton-ap-obstruction}.

\begin{figure}[tbp]
\centering
\begin{tikzpicture}[x=1cm,y=1cm]
  \node[horn point,draw=hornblue,fill=hornblue!8] (B1a) at (-1.75,.55) {$a$};
  \node[horn point] (B1e) at (0,1.25) {$x_\varepsilon$};
  \node[horn point] (B10) at (-.78,0) {$x_0$};
  \node[horn point] (B11) at (.78,0) {$x_1$};
  \foreach \v in {B1e,B10,B11}
    \draw[densely dotted,line width=.5pt,draw=hornblue!75] (B1a) -- (\v);
  \draw[line width=.65pt,draw=horngray]
    (B1e) -- node[horn label,above left=-1pt,fill=white,inner sep=1pt] {$C_0$} (B10);
  \draw[line width=.65pt,draw=horngray]
    (B1e) -- node[horn label,above right=-1pt,fill=white,inner sep=1pt] {$C_1$} (B11);
  \draw[densely dotted,line width=.55pt,draw=horngray]
    (B10) to[bend right=22] node[horn label,below] {$D$} (B11);
  \node[horn label,text=hornblue,anchor=south] at (B1a.north) {$L,I$};
  \node[horn label,anchor=south] at ($(B1e.north)+(0,.08)$) {$A_f,T$};
  \node[horn label,anchor=north] at ($(B10.south)+(0,-.08)$) {$A_{c_0}$};
  \node[horn label,anchor=north] at ($(B11.south)+(0,-.08)$) {$A_{c_1}$};

  \node[horn point] (B2e) at (3.65,1.25) {$x_\varepsilon$};
  \node[horn point] (B20) at (2.87,0) {$x_0$};
  \node[horn point] (B21) at (4.43,0) {$x_1$};
  \node[horn point,draw=hornaccent,fill=hornaccent!8] (B2b) at (5.40,.55) {$b$};
  \foreach \v in {B2e,B20,B21}
    \draw[densely dotted,line width=.5pt,draw=hornaccent!80] (B2b) -- (\v);
  \draw[line width=.65pt,draw=horngray]
    (B2e) -- node[horn label,above left=-1pt,fill=white,inner sep=1pt] {$C_0$} (B20);
  \draw[line width=.65pt,draw=horngray]
    (B2e) -- node[horn label,above right=-1pt,fill=white,inner sep=1pt] {$C_1$} (B21);
  \draw[densely dotted,line width=.55pt,draw=horngray]
    (B20) to[bend right=22] node[horn label,below] {$D$} (B21);
  \node[horn label,text=hornaccent,anchor=south] at (B2b.north) {$R$};
  \node[horn label,anchor=south] at ($(B2e.north)+(0,.08)$) {$A_f,T$};
  \node[horn label,anchor=north] at ($(B20.south)+(0,-.08)$) {$A_{c_0}$};
  \node[horn label,anchor=north] at ($(B21.south)+(0,-.08)$) {$A_{c_1}$};

  \node[font=\Large,text=horngray] at (6.55,.65) {$\Rightarrow$};

\draw [opacity=0.5, draw=hornaccent,-,line width=0.35mm,double distance = 0.2cm,line cap=round] plot [smooth,tension=0.4 ] coordinates {  (7.60,.55) (11.10,.55)  (8.57,0) };
\draw [opacity=0.5, draw=hornaccent,-,line width=0.35mm,double distance = 0.2cm,line cap=round] plot [smooth,tension=0.4 ] coordinates {  (7.60,.55) (11.10,.55)  (10.13,0) };
\draw [opacity=0.5, draw=hornaccent,-,line width=0.35mm,double distance = 0.2cm,line cap=round] plot [smooth,tension=0.4 ] coordinates {  (7.60,.55) (11.10,.55)  (9.35,1.25) };

  \node[horn point,draw=hornblue,fill=hornblue!8] (Aa) at (7.60,.55) {$a$};
  \node[horn point] (Ae) at (9.35,1.25) {$x_\varepsilon$};
  \node[horn point] (A0) at (8.57,0) {$x_0$};
  \node[horn point] (A1) at (10.13,0) {$x_1$};
  \node[horn point,draw=hornaccent,fill=hornaccent!8] (Ab) at (11.10,.55) {$b$};
  \foreach \v in {Ae,A0,A1}{
    \draw[densely dotted,line width=.5pt,draw=hornblue!75] (Aa) -- (\v);
    \draw[densely dotted,line width=.5pt,draw=hornaccent!80] (Ab) -- (\v);
  }
  \draw[line width=.65pt,draw=horngray]
    (Ae) -- node[horn label,above left=-1pt,fill=white,inner sep=1pt] {$C_0$} (A0);
  \draw[line width=.65pt,draw=horngray]
    (Ae) -- node[horn label,above right=-1pt,fill=white,inner sep=1pt] {$C_1$} (A1);
  \draw[densely dotted,line width=.55pt,draw=horngray]
    (A0) to[bend right=22] node[horn label,below] {$D$} (A1);
  \node[horn label,text=hornblue,anchor=south] at (Aa.north) {$L,I$};
  \node[horn label,text=hornaccent,anchor=south] at (Ab.north) {$R$};

  \node[horn label,anchor=south] at ($(Ae.north)+(0,.08)$) {$A_f,T$};
  \node[horn label,anchor=north] at ($(A0.south)+(0,-.08)$) {$A_{c_0}$};
  \node[horn label,anchor=north] at ($(A1.south)+(0,-.08)$) {$A_{c_1}$};

\end{tikzpicture}
\caption{The obstruction associated with a rejected binary tree
$t=f(c_0,c_1)$.  The two models share the tree-shaped structure
$\struct B_0$.  The first side adds the left pole $a$, together with
the states reachable at each node ($P_s$-facts not shown), while the
second adds the right pole
$b$.  In any amalgam, the seed clauses produce $Q(a,b,x_0)$ and
$Q(a,b,x_1)$; the recursive clause then produces
$Q(a,b,x_\varepsilon)$, and the $T$-mark at the root activates
\eqref{eq:tree-parse-final}.  Hence the diagram has no amalgam.  Dotted
edges represent padding atoms $D$; brown hyperedges mark the forced $Q$-atoms.
}
\label{fig:tree-automaton-ap-obstruction}
\end{figure}

Conversely, suppose that \(L(\mathcal A)=T_\Sigma\), and consider an arbitrary
amalgamation diagram \(\struct B_1,\struct B_2\) as in
\cref{lem:cross-q-tree-extraction}.  By
\cref{lem:complete-free-union,lem:tree-acc-complete}, their union
\(\struct U\) satisfies \(\Phi_{\mathrm{acc}}\).  Form \(\struct U^*\) by the
\(Q\)-closure from \cref{lem:cross-q-tree-extraction}.  Since the domain and
signature are finite, this process reaches a fixed point.

We claim that \eqref{eq:tree-parse-final} does not fire in \(\struct U^*\).
It cannot fire on a \(Q\)-tuple already present in \(\struct U\): the
complete clause \eqref{eq:tree-parse-final} holds in \(\struct U\) by
\cref{lem:complete-free-union,lem:tree-acc-complete}.  If it fired on a newly
added tuple \(Q(a,b,c)\), then \cref{lem:cross-q-tree-extraction} would supply
a tree \(t\) and a map \(h\) of all its positions into \(B_0\), with
\(h(\varepsilon)=c\).  Choose an accepting run of \(\mathcal A\) on \(t\).
All atoms of the corresponding instance of \(\lambda_t(a,\bar h)\) hold in
the side containing \(a\). The map \(h\) need not be injective: the
run-to-entailment direction is a universal implication and remains valid
when its variables are assigned equal values. Applying that direction of
\cref{lem:tree-horn-simulation} inside that side derives
\(P_s(a,c)\) for a final state \(s\).  Since \(T(c)\) belongs to the common
substructure, \eqref{eq:tree-acc-final} contradicts the assumption that this
side is a model.

The clauses of \(\Phi_{\mathrm{acc}}\) do not mention \(Q\), so they
remain true in \(\struct U^*\). The seed and recursive parsing clauses
hold at the fixed point, and the final parsing constraint holds by the
preceding argument. Hence \(\struct U^*\models\Phi_{\mathcal A}\).  The only tuples
added to \(\struct U\) are \(Q\)-tuples whose poles lie in opposite sides, by
\cref{lem:cross-q-tree-extraction}.  Hence the inclusions of
\(\struct B_1,\struct B_2\) into \(\struct U^*\) are embeddings and introduce
no new identification.  Thus \(\struct U^*\) is a strong amalgam.
\end{proof}

\begin{proof}[Proof of \cref{thm:bounded-arity-hardness}]
The map \(\mathcal A\mapsto\Phi_{\mathcal A}\) is computable in polynomial
time, produces a universal Horn sentence, and uses relation symbols
of arity at most three.  By
\cref{lem:tree-automaton-ap-equivalence}, it maps universal automata to
sentences with SAP and non-universal automata to sentences without AP.
Universality for nondeterministic finite tree automata over ranked alphabets
of rank at most two is
\(\ComplexityClass{ExpTime}\)-hard by~\cite{SeidlTreeAutomata} and the rank
reduction in \cref{sec:tree-universality}. This proves the lower bound.%
\end{proof}

\section{Binary exponential corridor games}\label{sec:corridor-games}

The second reduction uses finite winning strategy trees for a game on
rows of width $2^n$, where $n$ is part of the input in unary. This extends
the exponential corridor-tiling problem used for the earlier
\ComplexityClass{ExpSpace} lower bound~\cite{RydvalICALP2024} by allowing
existential and universal choices. Corridor-tiling games go back to
Chlebus~\cite{ChlebusDominoTiling}. We give the full proof of the precise
normal form needed for Horn amalgamation.

\subsection{The game and its winning region}

We first specify the source problem.  A \emph{binary exponential corridor
game} is a tuple
\[
 \mathcal G=(n,T,H,V_0,V_1,
 \mathcal B_\ell,\mathcal B_r,\mathcal B_{\mathrm{in}},\mathcal B_{\mathrm{EP}},\mathcal B_{\mathrm{UP}},\mathcal B_{\mathrm{win}}),
\]
where $n\geq1$ is given in unary, $T$ is a finite tile set,
$H,V_0,V_1\subseteq T^2$, and the remaining entries are subsets of $T$.
We require $\mathcal B_{\mathrm{in}}\subseteq \mathcal B_\ell$.
Set $N=2^n$.  A \emph{row} is a tuple $\rho\in T^N$.  It is \emph{legal} if
$
 \rho(0)\in \mathcal B_\ell$, $\rho(N-1)\in \mathcal B_r$, and $(\rho(j),\rho(j+1))\in H$ for all $j<N-1$.
It is \emph{initial} if it is legal and its first tile belongs to $\mathcal B_{\mathrm{in}}$.
Independently, it is an \emph{Existential Player (EP) row}, a
\emph{Universal Player (UP) row}, or a \emph{winning row} depending on whether
its first tile belongs to $\mathcal B_{\mathrm{EP}},\mathcal B_{\mathrm{UP}}$, or
$\mathcal B_{\mathrm{win}}$, respectively.
We restrict ourselves to the following normal form:
\begin{enumerate}[label=\textup{(NF\arabic*)}]
 \item \label{item:NF1} there is a unique initial row;
 \item \label{item:NF2} $\mathcal B_{\mathrm{EP}},\mathcal B_{\mathrm{UP}},\mathcal B_{\mathrm{win}}$ partition $\mathcal B_\ell$;
 \item \label{item:NF3} for every non-winning legal row $\rho$ and every $i\in\{0,1\}$,
 there is a unique legal row $\rho'$ such that $ (\rho(j),\rho'(j))\in V_i$ for all $j<N$; we write  $\rho\rightarrow_i\rho'$.
\end{enumerate}
The game starts at the initial row.  At an EP row, EP chooses
$i\in\{0,1\}$; at a UP row, UP chooses $i\in\{0,1\}$.  EP wins by reaching a
winning row, and UP wins infinite plays.

A \emph{finite winning strategy tree} for EP from a legal row $\rho$ is
a finite tree of legal rows rooted at $\rho$. Its leaves are winning rows
and its internal nodes are non-winning rows, with one $i$-successor at each
EP node and both the $0$- and the $1$-successor at each UP node.
Define sets of legal rows $W_r$ inductively.  Let $W_0$ consist of all
legal winning rows, and, for $r\geq 0$, put
\begin{align}
 W_{r+1}=W_r
 &\mathbin{\cup}
 \{\rho\mid \rho\text{ is an EP row and, for some }i\in\{0,1\},
                  \ \rho\to_i\rho_i\text{ with }\rho_i\in W_r\}
 \notag\\
 &\mathbin{\cup}
 \{\rho\mid \rho\text{ is a UP row and, for both }i\in\{0,1\},
                  \ \rho\to_i\rho_i\text{ with }\rho_i\in W_r\}.
 \label{eq:corridor-winning-stages}
\end{align}
The \emph{winning region} is defined as $W\coloneqq\bigcup_{r\geq0}W_r$.

Induction on $r$ shows that a row belongs to $W_r$ exactly when EP has a
finite winning strategy tree of height at most $r$ from that row, where
height is measured in edges. Moreover, every winning strategy yields a
finite one: unfold it and
stop each branch when it first reaches a winning row. The resulting tree
is finitely branching and has no infinite branch, because UP wins every
infinite play. By K\"onig's lemma, it is finite. Thus EP has a winning
strategy if and only if the initial row belongs to $W$.

\subsection{Complexity of the normal form}

A closely related formulation in terms of finite tiling trees with
exponentially wide rows appears in
\cite[extended version, Definition~5.4 and Lemma~5.5]{AxelssonHagueKreutzerLangeLatte}; see
also \cite[extended version, Proposition~9]{ArtaleEtAlSafety}. We require
the explicitly labelled unique-successor condition~\ref{item:NF3}.
The following lemma establishes completeness for this exact variant.

\begin{lemma}\label{lem:binary-corridor-hardness}
Deciding whether EP has a winning strategy in a binary exponential corridor
game satisfying \ref{item:NF1}--\ref{item:NF3} is
\ComplexityClass{2ExpTime}-complete.
\end{lemma}

We use one-tape alternating Turing machines
\[
 M=(Q_{\exists},Q_{\forall},Q_{\mathrm{acc}},Q_{\mathrm{rej}},
     \Sigma,\Gamma,\square,\delta,q_{\mathrm{init}}),
\]
where the four sets $Q_{\exists},Q_{\forall},Q_{\mathrm{acc}},Q_{\mathrm{rej}}$ of \emph{states} form a partition of $Q$, the input
alphabet $\Sigma$ is contained in the tape alphabet $\Gamma$,
$\square\in\Gamma\setminus\Sigma$ is the blank symbol, $q_{\mathrm{init}}\in Q$ is the \emph{initial state}, and we have a \emph{transition function}
\[
 \delta\colon
 (Q_{\exists}\cup Q_{\forall})\times\Gamma
 \longrightarrow  2^{(Q\times\Gamma\times\{-1,0,1\})}.
\]
We assume that $\delta(q,a)$ is nonempty for every
$q\in Q_{\exists}\cup Q_{\forall}$ and $a\in\Gamma$; accepting and rejecting
states have no outgoing transitions.
The tape is indexed by the non-negative integers. A move to the left from cell
$0$ leaves the head at cell $0$.

A \emph{configuration} consists of a state,
a head position, and a tape content which is blank at all but finitely many
positions. If $(q',b,d)\in\delta(q,a)$, then the corresponding successor is
obtained by writing $b$ in place of the symbol $a$ scanned by the head,
changing the state to $q'$, and moving the head by $d$, subject to the
convention at cell $0$. On input $w\in\Sigma^*$, the initial configuration
$C_{\mathrm{init}}(w)$ has
state $q_{\mathrm{init}}$, has $w$ written from cell $0$ onwards and
blanks thereafter, and has its head at cell $0$.

We use the acceptance-game semantics. A configuration in $Q_{\exists}$ is
owned by EP and one in $Q_{\forall}$ by UP; its owner chooses one of its
successor configurations. Reaching a configuration in $Q_{\mathrm{acc}}$ is
winning for EP, whereas reaching one in $Q_{\mathrm{rej}}$, or producing an
infinite play, is winning for UP. Thus an input is accepted precisely when EP
has a winning strategy from its initial configuration.

The machine uses space at most $s(m)$ on inputs of length $m$ if the head
position is smaller than $s(m)$ in every configuration reachable from the
initial configuration on such an input.
The proof below relies on the alternation theorem of
Chandra, Kozen, and Stockmeyer~\cite{ChandraKozenStockmeyerAlternation} in the form
$
 \ComplexityClass{AExpSpace}=\ComplexityClass{2ExpTime}.
$

\begin{proof}[Proof of \cref{lem:binary-corridor-hardness}]
\emph{Upper bound.} Let $\ell_{\mathcal G}$ be the size of the encoding of the game and put
$N=2^n$. There are at
most
 $
 N_{\mathrm{rows}}\coloneqq |T|^N=|T|^{2^n}
 \leq 2^{2^{O(\ell_{\mathcal G})}}
$
rows. Given one or two explicitly represented rows, legality and the relations
$\mathord\to_i$ can be tested in time
$N\cdot \operatorname{poly}(\ell_{\mathcal G})$.  By the induction following
\eqref{eq:corridor-winning-stages}, a row belongs to $W_r$ exactly when it
admits a winning strategy tree of height at most $r$.  The increasing sequence
$W_0\subseteq W_1\subseteq\cdots$ stabilizes after at most
$N_{\mathrm{rows}}$ stages.  Enumerating the rows and row pairs and computing
this sequence therefore takes
$N_{\mathrm{rows}}^{O(1)}\cdot N \cdot\operatorname{poly}(\ell_{\mathcal G})
=2^{2^{O(\ell_{\mathcal G})}}$ time. This proves membership in
$\ComplexityClass{2ExpTime}$.\footnote{The same enumeration can also verify conditions~\ref{item:NF1}--\ref{item:NF3} within this bound if membership in the normal form
is regarded as part of the input language rather than as a promise.}

\emph{Reduction from alternating machines.} For hardness, let $L\in\ComplexityClass{2ExpTime}$. By the alternation theorem,
there are a polynomial \(p\) and an alternating Turing machine \(M\)
which decides \(L\) using space at most \(2^{p(m)}\) on inputs of
length \(m\). We choose the normal form of \(M\) described below
independently of the input. Replacing \(p\) uniformly by
\(m\mapsto p(m)+m+2\), we may also assume that \(p(m)\geq m+2\),
and hence \(m+1<2^{p(m)}\), for every input length \(m\).
Given an input word \(w\), with \(m=|w|\), we construct in polynomial
time a binary exponential corridor game
$
  \mathcal G_{M,w}
$
of width
$
  N=2^{p(m)}
$
such that
\[
  w\in L
  \quad\Longleftrightarrow\quad
  \text{EP has a winning strategy in }\mathcal G_{M,w}.
\]
The remainder of the proof defines \(\mathcal G_{M,w}\) and verifies
the equivalence.

\emph{Machine normal form.} We bring $M$ into a binary-branching normal
form by changes that are independent of the input word.
For every non-halting state \(q\) and every tape symbol \(a\in\Gamma\), choose a finite full binary decision tree whose leaf labels are exactly
the transitions in \(\delta(q,a)\), each occurring at least once.
The root represents the original state scanning
that symbol, and every other internal vertex is represented by a fresh
auxiliary state with the same owner as the original state. Label the two edges
leaving every internal vertex by $0$ and $1$. An edge leading to another
internal vertex changes only the state, whereas an edge leading to a leaf
performs the original transition labelling that leaf. Duplicate leaves if
necessary; in particular, if only one original transition is available, label
both leaves by it. If an auxiliary state is paired with a scanned symbol for
which it was not intended, let both labelled transitions lead, without
changing the tape or head, to a fresh rejecting state. Such configurations
are unreachable.
For an existential state, the decision tree replaces a finite disjunction of
successors by iterated binary disjunctions; for a universal state, it replaces
a finite conjunction by iterated binary conjunctions. Hence acceptance
is unchanged. The auxiliary transitions use no additional tape cells,
and every non-halting configuration now has exactly one successor  under each
label $i\in\{0,1\}$.

We also arrange that the initial state $q_{\mathrm{init}}$ does not occur as
the target state of any transition. Rename the old initial state and add a
fresh state used only at the start, with the same owner (existential or
universal) or the same halting status (accepting or rejecting) and, when
non-halting, for every scanned symbol $a\in\Gamma$, the same two labelled
transitions as the renamed old initial state has on $a$.
Now fix an input word $w$ of length $m$, and set
$n\coloneqq p(m)$ and $N\coloneqq 2^n$.

 An \emph{$N$-cell configuration} is a configuration
whose head position lies in $\{0,\ldots,N-1\}$ and whose tape is blank
outside this interval.  Write $\mathcal C_N$ for the set of all $N$-cell
configurations.
On $\mathcal C_N$, we  \emph{totalize} both labelled successor
maps: if a prescribed move would leave the interval $\{0,\ldots,N-1\}$, its
write and  state update are performed but the head remains at the
corresponding boundary cell. At the left boundary this agrees with our tape
convention. At the right boundary it changes no transition reachable from the
initial configuration, because $M$ uses at most $N$ tape cells on input $w$.
Hence totalization does not change the winner of the acceptance game
from $C_{\mathrm{init}}(w)$.

We represent every
rejecting halting configuration by a UP row which is not marked as winning,
and give it itself as its successor under both labels. Thus reaching a
rejecting configuration of $M$ is simulated by an infinite self-looping play,
which is winning for UP. Consequently, for every non-accepting
$C\in\mathcal C_N$ and every $i\in\{0,1\}$, there is a unique successor under
label $i$; we denote it by $\operatorname{succ}_i(C)$. For a rejecting
configuration, this convention gives $\operatorname{succ}_i(C)=C$.
Set
\[
 \mathcal S_w\coloneqq
 \{C\in\mathcal C_N \mid \text{the state of $C$ is not $q_{\mathrm{init}}$}\}
 \cup\{C_{\mathrm{init}}(w)\}.
\]
Since no transition enters $q_{\mathrm{init}}$, the set $\mathcal S_w$
contains the initial configuration and is closed under both successor maps on
its non-accepting configurations.

We use the following standard claim, spelling out the
part needed for condition~\ref{item:NF3}. Given $M$ and an input word $w$ of
length $m$, one can construct in polynomial time a tile set $T$, relations $H,V_0,V_1\subseteq T^2$, and sets
$\mathcal B_\ell,\mathcal B_r,\mathcal B_{\mathrm{in}},
\mathcal B_{\mathrm{EP}},\mathcal B_{\mathrm{UP}},
\mathcal B_{\mathrm{win}}\subseteq T$,
with the already fixed parameter $n=p(m)$ written in unary,
together with a bijection
$
 \operatorname{enc}\colon
 \mathcal S_w
 \longrightarrow \{\text{legal rows}\},
$
such that:
\begin{itemize}
    \item the initial configuration is mapped to the unique initial row;
    \item the
first tile records whether the configuration is existential, universal,
accepting, or rejecting; and
    \item for every non-accepting
$C\in\mathcal S_w$ and every $C'\in\mathcal S_w$,
\begin{align}
 \bigl(\operatorname{enc}(C)(j),\operatorname{enc}(C')(j)\bigr)\in V_i
 \text{ for all }j<N
 \quad\Longleftrightarrow\quad
 C'=\operatorname{succ}_i(C). \label{eq:succ}
\end{align}
\end{itemize}

\emph{Horizontal encoding.} To prove the claim, encode a configuration as a word of length $N$ over the
tape alphabet, with one distinguished head position. The tile at position $j$
records the three-cell window consisting of positions $j-1,j,j+1$, with
entries in $\Gamma\cup\{\triangleright,\triangleleft\}$. At the two
boundaries, a missing neighbouring cell is replaced by the left endmarker
$\triangleright$ or the right endmarker $\triangleleft$, respectively.
The centre of every tile is a genuine tape symbol, and membership
in $\mathcal B_\ell$ and $\mathcal B_r$ requires the corresponding boundary pattern. The tile
also records the current state, copied into every tile, and a
\emph{head-position tag} from
$\{\mathsf{before},\mathsf{head},\mathsf{after}\}$. If the unique head is at
position $h$, the intended tag at position $j$ is $\mathsf{before}$ when
$j<h$, $\mathsf{head}$ when $j=h$, and $\mathsf{after}$ when $j>h$.
Thus the intended sequence of tags has the
form
$
 \mathsf{before}^{*}\,\mathsf{head}\,\mathsf{after}^{*}.
$
The relation $H$ enforces this pattern locally. Besides requiring successive
windows to overlap exactly and the copied states to agree, it permits
only the following pairs of tags on horizontally adjacent tiles:
\[
 (\mathsf{before},\mathsf{before}),\qquad
 (\mathsf{before},\mathsf{head}),\qquad
 (\mathsf{head},\mathsf{after}),\qquad
 (\mathsf{after},\mathsf{after}).
\]
Moreover, $\mathcal B_\ell$ contains only tiles whose tag is
$\mathsf{before}$ or $\mathsf{head}$, whereas $\mathcal B_r$ contains only tiles whose
tag is $\mathsf{head}$ or $\mathsf{after}$. The allowed adjacent pairs forbid
two head tags and also forbid passing directly from $\mathsf{before}$ to
$\mathsf{after}$. The right boundary excludes an all-$\mathsf{before}$ row,
and the left boundary excludes an all-$\mathsf{after}$ row. Hence every legal
row has exactly one distinguished head position.

Add a deterministic horizontal position component with $O(m+1)$ states which
labels cells $0,\ldots,m$ individually and then remains in a tail state; $H$
determines this component throughout every legal row, and $\mathcal B_\ell$ fixes its
initial value. Cells $0,\ldots,m-1$ are the input positions, while cell $m$
marks the start of the blank tail in the initial configuration. The
inequality $m+1<N$ ensures that all these distinguished positions fit into
the row. The position component itself places no restriction on tape symbols
outside the initial state. In the fresh initial state $q_{\mathrm{init}}$, restrict the permitted
tiles so that the first $m$ cells spell $w$, the head is in its initial
position, and all remaining cells are blank.  For every other state,
all well-formed configurations are permitted.  Exact window overlap and the
boundary conditions now show that the legal rows are precisely the
encodings of the configurations in $\mathcal S_w$.  This proves the existence of the bijection  $\operatorname{enc}$, and it also makes the encoding of $C_{\mathrm{init}}(w)$ the
unique initial row.  Let $\mathcal B_{\mathrm{in}}$ consist of the left-boundary tiles
carrying $q_{\mathrm{init}}$.  The leftmost tile carries the copied control
state, so the existential, universal, accepting, and rejecting states
define a partition
$\mathcal B_{\mathrm{EP}},\mathcal B_{\mathrm{UP}},\mathcal B_{\mathrm{win}}$ of $\mathcal B_\ell$, where rejecting
states are placed in $\mathcal B_{\mathrm{UP}}$.

For an interior head position, the three tiles centred immediately around the head are illustrated
in~\cref{fig:corridor-row-encoding};
the deterministic position component is omitted.
\begin{figure}[ht]
\centering
\begin{tikzpicture}[x=1cm,y=1cm]
  \node[horn math,minimum width=3.55cm,minimum height=1.55cm,
        font=\scriptsize] (lefttile) at (0.85,.82)
    {$\rho(h-1)$\\[2pt]
     window $(a_{h-2},a_{h-1},a_h)$\\[1pt]
     state $q$, tag $\mathsf{before}$};
  \node[horn math,draw=hornaccent,fill=hornwarmfill,
        minimum width=3.55cm,minimum height=1.55cm,
        font=\scriptsize] (headtile) at (5.85,.82)
    {$\rho(h)$\\[2pt]
     window $(a_{h-1},a_h,a_{h+1})$\\[1pt]
     state $q$, tag $\mathsf{head}$};
  \node[horn math,minimum width=3.55cm,minimum height=1.55cm,
        font=\scriptsize] (righttile) at (10.85,.82)
    {$\rho(h+1)$\\[2pt]
     window $(a_h,a_{h+1},a_{h+2})$\\[1pt]
     state $q$, tag $\mathsf{after}$};

  \draw[horn relation] (lefttile.east) --
    node[horn label,above] {$H$} (headtile.west);
  \draw[horn relation] (headtile.east) --
    node[horn label,above] {$H$} (righttile.west);
\end{tikzpicture}
\caption{Assuming \(1\leq h\leq N-2\), three horizontally consecutive
tiles \(\rho(h-1),\rho(h),\rho(h+1)\), where \(\rho\in T^N\) is the
legal row encoding a configuration whose head is at position \(h\).
Here \(a_{-1}=\triangleright\) and \(a_N=\triangleleft\).
At a boundary head position, the missing neighbouring tape symbol is
replaced by the corresponding endmarker. A tile stores
an overlapping three-cell tape window, the copied state $q$, and its
head-position tag; $H$ checks adjacent overlap and horizontal consistency.}
\label{fig:corridor-row-encoding}
\end{figure}

\emph{Vertical compatibility.} Since every transition moves the head by
at most one cell, its effect can be checked locally using the scanned cell
and its two neighbours. We define $V_i$ to consist precisely of the pairs of source and target
tiles satisfying the conditions described below.
The relation compares tiles at the same position and requires their
deterministic position components to agree. On
the head-position tags it never permits either of the pairs
$
 (\mathsf{before},\mathsf{after}) $ and $
 (\mathsf{after},\mathsf{before}).$
Suppose that legal source and target rows have their heads at positions $h$
and $h'$, respectively.  If $h'\leq h-2$, then position $h-1$ has vertical
tag pair $(\mathsf{before},\mathsf{after})$; if $h'\geq h+2$, then position
$h+1$ has vertical tag pair $(\mathsf{after},\mathsf{before})$.  The two
exclusions therefore imply
$
 |h-h'|\leq 1.  $

For an accepting source state, $V_i$ contains no outgoing pair. Accepting
rows are winning, so condition~\ref{item:NF3} imposes no successor
requirement on them. For a rejecting source state, we let $V_i$ require equality of the source and
target data, thereby implementing the identity successor fixed above.
Suppose henceforth that the source state is non-halting.  At the source-head position, \(V_i\) sees the current state and the tape symbol under the head, and hence determines the unique totalized transition with label \(i\).  Let
$d\in\{-1,0,1\}$ be its  head displacement, where a prescribed
outward move at a boundary has displacement $0$.  At the
source-head position, $V_i$ requires the target tag to be
$
 \mathsf{after}$, $\mathsf{head}$, or
 $\mathsf{before}$
according as $d=-1,0,1$, respectively.  Together
with~$ |h-h'|\leq 1$, this forces
$h'=h+d$.

At the source-head position, $V_i$ also checks the symbol written there and
the new state.  At every other position it requires the centre tape symbol to
remain unchanged.  Exact window overlap assembles these centre-symbol
requirements into the target tape word, while horizontal compatibility in
the target row propagates the copied new state through the entire row.  The
boundary endmarkers make the effective displacement visible at the two ends.
Thus these local conditions implement the totalized successor map defined
above.
Consequently $V_i$ need not be the graph of a function on individual tiles,
but it induces the graph of the function
$C\mapsto\operatorname{succ}_i(C)$ on legal rows. This proves~\eqref{eq:succ}
and condition~\ref{item:NF3} follows: since $\operatorname{enc}$ is
onto the legal rows, every non-winning legal row is
$\operatorname{enc}(C)$ for a unique non-accepting $C\in\mathcal S_w$; closure
of $\mathcal S_w$ gives $\operatorname{succ}_i(C)\in\mathcal S_w$, and
\eqref{eq:succ} says that
$\operatorname{enc}(\operatorname{succ}_i(C))$ is the unique legal row
pointwise $V_i$-related to $\operatorname{enc}(C)$.

\emph{Size and correctness.} The tile set and all compatibility relations of
\(\mathcal G_{M,w}\) have size polynomial in \(|M|+|w|\), while its
width is specified succinctly by \(n=p(|w|)\), meaning that the actual
width is \(2^n\).  Hence
$
  w\longmapsto\mathcal G_{M,w}
$
is a polynomial-time many-one reduction. Under the bijection
\(\operatorname{enc}\), every labelled transition from a non-halting
configuration corresponds to the unique transition between the associated
legal rows. Accepting configurations correspond to winning rows, while
rejecting configurations are replaced by UP-owned identity loops. This
modification preserves the winner. Therefore
\[
  w\in L
  \iff M\text{ accepts }w
  \iff \text{EP wins }\mathcal G_{M,w}.
\]
Since \(L\) was an arbitrary language in
\(\ComplexityClass{2ExpTime}\), the binary exponential corridor game
is \(\ComplexityClass{2ExpTime}\)-hard.
\end{proof}

\section{The lower bound for unbounded arity}\label{sec:unbounded-arity}

We use the normal form from \cref{sec:corridor-games} to translate the game rules into a Horn sentence.
Binary branching lets
one represent an EP choice by choosing one of two Horn rules, and a UP choice
by one Horn rule containing exactly two successors.  Moreover, the label $i$
must determine a unique legal successor from every non-winning legal row:
at a UP row, the Horn rule
certifies one winning successor for each label, which covers all possible
moves precisely because of this uniqueness. A vertical step check alone
verifies only pointwise compatibility; legality of its endpoints is established separately in the
strategy-extraction proof. Marking ownership and acceptance in
the first tile lets the row scanner select the appropriate rule. In the
games produced by \cref{lem:binary-corridor-hardness}, the two identity
successors of a rejecting row turn rejection into an infinite play, which is
winning for UP.

\subsection{The signature and binary addresses}

We now translate a game $\mathcal G$ into a Horn sentence.  The construction is
a streamlined version of the binary-counter mechanism from
\cite[proof of Theorem~6]{RydvalICALP2024}.  Besides unary
symbols $L,R$ and binary symbols $E,D$, the  signature contains an
$(n+1)$-ary symbol $T_\alpha$ for every $\alpha\in T$.  An atom
$
 T_\alpha(c_1,\ldots,c_n,p)
$
says that the row represented by $p$ contains $\alpha$ at address
$(c_1,\ldots,c_n)$.  The binary relation $D$ distinguishes the two elements
used as bits.  We include the complete Horn clauses
\begin{align}
 D(x,x')\wedge D(x',x)&\Rightarrow\bot,
 \label{eq:corridor-bit-asymmetry}\\
 T_\alpha(\bar c,p)\wedge T_\beta(\bar c,p)&\Rightarrow\bot
 \qquad(\alpha\ne\beta).
 \label{eq:corridor-tile-functionality}
\end{align}

For a list $\bar z$ of variables, put
\[
 \operatorname{PAD}(y_1,y_2;\bar z)\coloneqq
 L(y_1)\wedge R(y_2)\wedge
 \bigwedge\nolimits_{z\in\bar z}\bigl(E(y_1,z)\wedge E(y_2,z)\bigr)
 \wedge\bigwedge\nolimits_{z,z'\in\bar z}E(z,z').
\]
The last conjunction includes the loops $E(z,z)$, which we retain for
uniformity. These diagonal conjuncts are not needed for completeness.
Even if they were omitted, a retained atom $E(z,z')$ could still require a
ground loop when an assignment identifies the distinct variables $z,z'$.
The contribution of $\operatorname{PAD}(y_1,y_2;\bar z)$ to the premise
Gaifman graph contains every edge among its variables except
$(y_1,y_2)$.
Fix variables $b_0,b_1$ for the two bits, and write
\(\mathbf 0=(b_0,\ldots,b_0)\) and
\(\mathbf 1=(b_1,\ldots,b_1)\).  For $j\in[n]$ and
$\bar z=(z_1,\ldots,z_{j-1})$, define
\begin{align*}
 \ell_j(\bar z)&=(z_1,\ldots,z_{j-1},b_0,
                   \underbrace{b_1,\ldots,b_1}_{n-j}),\\
 r_j(\bar z)&=(z_1,\ldots,z_{j-1},b_1,
                   \underbrace{b_0,\ldots,b_0}_{n-j}).
\end{align*}
These are the two sides of a binary increment in which
position $j$ flips from $0$ to $1$.

\subsection{Row and transition scanners}

For
$X\in\{\mathrm{in},\mathrm{EP},\mathrm{UP},\mathrm{win}\}$ introduce an
$(n+5)$-ary scan relation $S_X$ and a $5$-ary relation
$\operatorname{Row}_X$.  For every $\alpha\in \mathcal B_X$ include the seed rule
\begin{equation}
 \begin{split}
 D(b_0,b_1)\wedge T_\alpha(\mathbf0,p)
 \wedge\operatorname{PAD}(y_1,y_2;p,b_0,b_1)
 \Rightarrow
 S_X(y_1,y_2,p,\mathbf0,b_0,b_1).
 \end{split}
 \label{eq:corridor-row-seed}
\end{equation}
For every $j\in[n]$ and $(\alpha,\beta)\in H$ include
\begin{equation}
 \begin{split}
 S_X(y_1,y_2,p,\ell_j(\bar z),b_0,b_1)
 \wedge T_\alpha(\ell_j(\bar z),p)
 \wedge T_\beta(r_j(\bar z),p)\\
 \Rightarrow S_X(y_1,y_2,p,r_j(\bar z),b_0,b_1),
 \end{split}
 \label{eq:corridor-row-step}
\end{equation}
and, for every $\alpha\in \mathcal B_r$, include
\begin{equation}
 S_X(y_1,y_2,p,\mathbf1,b_0,b_1)
 \wedge T_\alpha(\mathbf1,p)
 \Rightarrow\operatorname{Row}_X(y_1,y_2,p,b_0,b_1).
 \label{eq:corridor-row-finish}
\end{equation}

The transition scanner is analogous.  For $i\in\{0,1\}$ introduce an
$(n+6)$-ary relation $S_i^\to$ and a $6$-ary relation
$\operatorname{Step}_i$.  For $(\alpha,\beta)\in V_i$ include
\begin{equation}
 \begin{split}
 D(b_0,b_1)\wedge T_\alpha(\mathbf0,p)
 \wedge T_\beta(\mathbf0,p')
 \wedge\operatorname{PAD}(y_1,y_2;p,p',b_0,b_1)\\
 \Rightarrow S_i^\to(y_1,y_2,p,p',\mathbf0,b_0,b_1),
 \end{split}
 \label{eq:corridor-edge-seed}
\end{equation}
and, for every $j\in[n]$ and $(\alpha,\beta)\in V_i$, include
\begin{equation}
 \begin{split}
 S_i^\to(y_1,y_2,p,p',\ell_j(\bar z),b_0,b_1)
 \wedge T_\alpha(r_j(\bar z),p)
 \wedge T_\beta(r_j(\bar z),p')\\
 \Rightarrow S_i^\to(y_1,y_2,p,p',r_j(\bar z),b_0,b_1).
 \end{split}
 \label{eq:corridor-edge-step}
\end{equation}
Finally, include
\begin{equation}
 S_i^\to(y_1,y_2,p,p',\mathbf1,b_0,b_1)
 \Rightarrow\operatorname{Step}_i(y_1,y_2,p,p',b_0,b_1).
 \label{eq:corridor-edge-finish}
\end{equation}

\subsection{Winning strategies}

It remains to recognize winning strategy trees.  Introduce a $5$-ary relation
$\operatorname{Win}$.  The following rules mirror the induction defining the sets $W_r$
in~\eqref{eq:corridor-winning-stages}:
\begin{equation}
 \operatorname{Row}_{\mathrm{win}}(y_1,y_2,p,b_0,b_1)
 \Rightarrow \operatorname{Win}(y_1,y_2,p,b_0,b_1).
 \label{eq:corridor-win-leaf}
\end{equation}
For each $i\in\{0,1\}$, include
\begin{multline}
 \operatorname{Row}_{\mathrm{EP}}(y_1,y_2,p,b_0,b_1)
 \wedge\operatorname{Step}_i(y_1,y_2,p,p',b_0,b_1)
 \wedge \operatorname{Win}(y_1,y_2,p',b_0,b_1)\\
 \wedge\operatorname{PAD}(y_1,y_2;p,p',b_0,b_1)
 \Rightarrow \operatorname{Win}(y_1,y_2,p,b_0,b_1).
 \label{eq:corridor-win-ep}
\end{multline}
For UP rows, include
\begin{multline}
 \operatorname{Row}_{\mathrm{UP}}(y_1,y_2,p,b_0,b_1)
 \wedge\operatorname{Step}_0(y_1,y_2,p,p_0,b_0,b_1)
 \wedge \operatorname{Win}(y_1,y_2,p_0,b_0,b_1)\\
 \wedge\operatorname{Step}_1(y_1,y_2,p,p_1,b_0,b_1)
 \wedge \operatorname{Win}(y_1,y_2,p_1,b_0,b_1)
 \wedge\operatorname{PAD}(y_1,y_2;p,p_0,p_1,b_0,b_1)\\
 \Rightarrow \operatorname{Win}(y_1,y_2,p,b_0,b_1).
 \label{eq:corridor-win-up}
\end{multline}
We retain the padding in \eqref{eq:corridor-win-ep} for uniformity, although
its \(\operatorname{Step}_i\)-atom already covers all six variables for
completeness and the extraction invariant. Removing it would strengthen the
clause on arbitrary structures. In \eqref{eq:corridor-win-up}, padding is
needed to cover the pair \((p_0,p_1)\).
Finally, include
\begin{equation}
 \operatorname{Row}_{\mathrm{in}}(y_1,y_2,p,b_0,b_1)
 \wedge \operatorname{Win}(y_1,y_2,p,b_0,b_1)\Rightarrow\bot.
 \label{eq:corridor-win-final}
\end{equation}
Let $\Phi_{\mathcal G}$ be the conjunction of all the displayed clauses.

\subsection{Correctness of the reduction}

For the next lemma, keep \(L,R,D,E,T_\alpha\) fixed, assuming
\eqref{eq:corridor-bit-asymmetry} and
\eqref{eq:corridor-tile-functionality}, and take closure under the scanner
rules \eqref{eq:corridor-row-seed}--\eqref{eq:corridor-edge-finish}.
For each scanner considered, fix its poles, row identifier(s), and bit
elements, and assume that its scan and output relations
(\(S_X,\operatorname{Row}_X\) or \(S_i^\to,\operatorname{Step}_i\)) are
initially empty at these parameters. Other parameter slices are unrestricted.

\begin{lemma}\label{lem:corridor-scanners}
The scanner closure has the following properties:
\begin{enumerate}[label=\textup{(\roman*)}]
 \item \label{item:funct1}
 \(\operatorname{Row}_X(y_1,y_2,p,b_0,b_1)\) is derived if and only if
 \[
  D(b_0,b_1)\wedge\operatorname{PAD}(y_1,y_2;p,b_0,b_1)
 \]
 holds and the \(T_\alpha\)-facts at \(p\) encode a legal row whose first
 tile belongs to \(\mathcal B_X\).
 \item \label{item:funct2}
 \(\operatorname{Step}_i(y_1,y_2,p,p',b_0,b_1)\) is derived if and only if
 \[
  D(b_0,b_1)\wedge\operatorname{PAD}(y_1,y_2;p,p',b_0,b_1)
 \]
 holds and the \(T_\alpha\)-facts at \(p,p'\) encode rows whose
 corresponding tiles are \(V_i\)-related.
\end{enumerate}
\end{lemma}

Tile functionality makes each encoded row unique.

\begin{proof}
Every recursive scanner rule preserves its parameter tuple, and finishing
rules preserve these parameters as well. Thus facts in other slices do not
affect the computation in the fixed slice. Any derived fact in an initially
empty slice therefore traces back to a seed with the same parameters, whose
premises include $D(b_0,b_1)$ and the relevant padding. If either condition
fails, no output fact is derived and the claimed equivalence is immediate.
Assume henceforth that both conditions hold.

Clause \eqref{eq:corridor-bit-asymmetry} and $D(b_0,b_1)$ imply
$b_0\ne b_1$. Write $\operatorname{bin}(j)\in\{b_0,b_1\}^n$ for the
binary address of $j$, with the most significant bit first. Every scan tuple
derived in the fixed parameter slice has a binary address: a seed starts at
$\mathbf0$, and a step
copies its prefix from the preceding scan tuple. Every non-maximal binary
address has a unique representation $\ell_k(\bar z)$, where $k$ is the
position of its rightmost $b_0$. Replacing it by $r_k(\bar z)$ increments
the address by one. No step applies at the maximal address.

Within the fixed row-scanner slice, induction on $j$ proves that
$S_X(y_1,y_2,p,\operatorname{bin}(j),b_0,b_1)$ belongs to the closure
if and only if there are tiles $\alpha_0,\ldots,\alpha_j$ such that
\begin{itemize}
    \item $\alpha_0\in \mathcal B_X,$
    \item $T_{\alpha_h}(\operatorname{bin}(h),p)$  holds for $0\leq h\leq j$,
    \item $(\alpha_h,\alpha_{h+1})\in H$ for $0\leq h<j$.
\end{itemize}
At address zero, only a seed can apply. At a positive address, the only
possible step comes from its immediate predecessor. Tile functionality
ensures that the tile tested there agrees with the final tile of the
previous prefix; the step therefore extends that prefix consistently.
Conversely, a compatible prefix supplies the premises of the unique
increment step. The finishing rule checks the right boundary, proving
\ref{item:funct1}.

For the transition scanner, fix both row identifiers as well as the poles
and bit elements. The analogous induction says that the scanner reaches
$\operatorname{bin}(j)$ precisely when both rows have tiles at every
address through $j$ and the corresponding tile pairs belong to $V_i$.
The finishing rule proves \ref{item:funct2}.
\end{proof}

The transition scanner does not check horizontal compatibility or boundary
conditions. Thus a $\operatorname{Step}_i$-fact need not connect legal
rows. If both endpoints are legal and the source is non-winning, however,
\cref{lem:corridor-scanners}\ref{item:funct2} and condition~\ref{item:NF3}
identify the target as the unique $i$-successor of the source.

All clauses above except the two scanner seed schemes are complete. Indeed,
each recursive or finishing scanner clause contains a scan atom involving all
of its variables. The winning-leaf and final clauses each contain an atom
involving all of their variables, while the EP and UP clauses are complete
because of their high-arity atoms and the displayed padding. A seed premise is
complete except for the single missing edge $(y_1,y_2)$.

\begin{lemma}\label{lem:corridor-cross-extraction}
Let $\struct B_1,\struct B_2\models\Phi_{\mathcal G}$ induce the same
substructure $\struct B_0$ on $B_1\cap B_2$, and let
$\struct U=\struct B_1\cup\struct B_2$.  Close $\struct U$ under all clauses
of $\Phi_{\mathcal G}$ with a positive head.  Every newly added tuple has its
first two entries in opposite sides outside $B_0$, and all remaining entries
in $B_0$.  If the premise of \eqref{eq:corridor-win-final} becomes true using
a newly added tuple, then the common $T_\alpha$-facts unfold to a finite
winning strategy tree for EP whose root is an initial row.
\end{lemma}

\begin{proof}
For the first assertion, we proceed by induction on the closure stage.

A new seed tuple cannot
have a pole in \(B_0\).  Indeed, suppose first that exactly one pole lies in
\(B_0\), and choose the side containing the other pole.  Every remaining
variable is joined by a padding atom to the outside pole, and hence lies in
that side.  If both poles lie in \(B_0\) and some remaining entry lies outside
\(B_0\), choose the side containing such an entry; the padding atoms between
all remaining entries put all of them in that side.  If all remaining entries
lie in \(B_0\), choose either side.  In every case each premise atom belongs
to the chosen side, using agreement of the induced structures for tuples contained in
\(B_0\).  Since the side satisfies the seed clause, the head would already
belong to \(\struct U\), a contradiction.  The two poles cannot lie outside
\(B_0\) on the same side for the same reason: each remaining variable is
joined to either pole.  Hence the poles lie in opposite sides outside
\(B_0\).  The padding atoms connecting every remaining entry to both poles
now force all remaining entries into \(B_0\).

In every non-seed rule, a new head requires a premise tuple outside
\(\struct U\): if every premise tuple belonged to \(\struct U\), completeness
would already put the head in \(\struct U\). The symbols
$L,R,D,E,T_\alpha$ never occur in a head, so their interpretations do not
change during closure. The new premise tuple must therefore belong to a
scan, row, step, or $\operatorname{Win}$ relation, each of which carries the poles in its
first two positions. It was added at an earlier stage, so the induction
hypothesis fixes the two cross poles. In a recursive or finishing scanner
rule, the scan premise contains every variable of the rule and, having
these cross poles, lies outside $\struct U$. Its earlier-stage invariant
places every data entry in $B_0$. The same argument applies to the row
premise in the winning-leaf rule. In an EP rule, the
$\operatorname{Step}_i$ premise also carries both cross poles and hence
lies outside $\struct U$; its earlier-stage invariant covers all six
variables of the rule. In a UP rule, the padding joins every data
variable to both poles, forcing its value into $B_0$. Thus every new head
satisfies the asserted invariant.

It remains to extract the strategy. Fix the two cross poles of a newly
derived tuple. No tuple
already in \(\struct U\) can contain both poles, because every relation tuple
of \(\struct U\) lies wholly in one side. Thus the slices of all scan, row,
step, and strategy relations carrying these poles are initially empty,
as required for \cref{lem:corridor-scanners}.
Since \eqref{eq:corridor-bit-asymmetry} and
\eqref{eq:corridor-tile-functionality} are complete, they hold in
\(\struct U\) by \cref{lem:complete-free-union}, and closure changes neither
\(D\) nor the \(T_\alpha\). The scanner relations
occur as heads only in
\eqref{eq:corridor-row-seed}--\eqref{eq:corridor-edge-finish}, whose
premises involve only scanner relations and the frozen relations
$L,R,D,E,T_\alpha$. Consequently, full positive closure and closure under
the scanner rules alone contain the same scanner facts. We use the
``only if'' directions of \cref{lem:corridor-scanners}, which also supply
the seed facts $D$ and padding.

We prove by induction on the first closure stage of each newly derived
$\operatorname{Win}$-tuple that its row identifier encodes a legal row and
that EP has a finite winning strategy tree from that row.
For \eqref{eq:corridor-win-leaf}, the premise
$\operatorname{Row}_{\mathrm{win}}$ supplies a legal winning row by
\cref{lem:corridor-scanners}\ref{item:funct1}; the strategy consists of one
leaf. For \eqref{eq:corridor-win-ep}, the
$\operatorname{Row}_{\mathrm{EP}}$ premise supplies a legal EP row. The
$\operatorname{Win}$-premise from an earlier closure stage supplies a legal child row and a finite winning
strategy from it. The $\operatorname{Step}_i$ premise gives pointwise
$V_i$-compatibility, so condition~\ref{item:NF3} identifies this child as
the parent's $i$-successor. Attach the child's strategy below the parent.
For \eqref{eq:corridor-win-up}, the same argument gives the two legal
successor rows and their finite winning strategies. Condition~\ref{item:NF3}
ensures that these account for both possible moves of UP. Attaching both
strategies proves the induction step, even if the two successors coincide.

Finally, if \eqref{eq:corridor-win-final} becomes true using a new tuple,
both of its premise tuples have the same cross poles and therefore neither
was present in $\struct U$. The newly derived
$\operatorname{Row}_{\mathrm{in}}$ premise marks the root of the extracted
strategy as initial.
\end{proof}

\begin{lemma}\label{lem:corridor-ap-equivalence}
For every binary exponential corridor game $\mathcal G$ in normal form,
\[
 \begin{aligned}
 \text{EP has no winning strategy in }\mathcal G
 &\quad\Longleftrightarrow\quad
 \fm(\Phi_{\mathcal G})\text{ has the SAP}\\
 &\quad\Longleftrightarrow\quad
 \fm(\Phi_{\mathcal G})\text{ has the AP}.
 \end{aligned}
\]
\end{lemma}

\begin{proof}
Suppose first that EP has a winning strategy.  By the finite-strategy
characterization following \eqref{eq:corridor-winning-stages}, EP then has a
finite winning strategy tree.  Put one element
$p_v$ into a common structure $\struct B_0$ for every node $v$ of this tree,
together with two further elements $b_0,b_1$, all chosen distinct.  Add $D(b_0,b_1)$ and all
$T_\alpha$-facts describing the row at $v$ on $p_v$.  Interpret $E$ as the
full relation on $B_0$, including all loops. Let $\struct B_1$ add an element
$a_1$ with $L(a_1)$ and $E(a_1,z)$ for all $z\in B_0$, and let
$\struct B_2$ add an element $a_2$ with $R(a_2)$ and $E(a_2,z)$ for all
$z\in B_0$. All scan and strategy relations are
empty in the two sides, and all unmentioned facts are absent. Since neither side contains both a left and a right
pole, no seed rule fires there; tile functionality and bit asymmetry hold by
construction.  Thus both sides are models.

The ``if'' directions of \cref{lem:corridor-scanners} apply with poles
$a_1,a_2$: $D(b_0,b_1)$ and all required padding facts are present in the
union. Thus the row and transition scanners derive all facts needed to
follow the given strategy tree. Induction from its leaves using
\eqref{eq:corridor-win-leaf}--\eqref{eq:corridor-win-up} derives $\operatorname{Win}$ at its
initial root, and \eqref{eq:corridor-win-final} derives $\bot$.  The diagram
therefore has no amalgam, since these derivations transfer along the
embeddings into any prospective amalgam, even one that identifies the poles.

Conversely, assume that EP has no winning strategy and take an arbitrary
amalgamation diagram.  By \cref{lem:complete-free-union}, its union satisfies
all complete clauses.  Close the union under the positive-head
clauses. The closure is finite. The bit-asymmetry and tile-functionality
constraints remain true, since their relations are unchanged, and all
positive-head clauses hold at the fixed point. Clause
\eqref{eq:corridor-win-final} cannot already fail in the union, since its
premise is complete and both sides are models.  It cannot be violated after
closure either: by \cref{lem:corridor-cross-extraction}, that would yield a
finite winning strategy for EP.  Hence the closure is a model of
$\Phi_{\mathcal G}$. The only added tuples have opposite outside poles, so the
inclusions of $\struct B_1$ and $\struct B_2$ into the closure are embeddings.
Since the closure has universe $B_1\cup B_2$, their images intersect precisely
in $B_0$. The closure is therefore a strong amalgam.
\end{proof}

\begin{proof}[Proof of \cref{thm:unbounded-arity-hardness}]
The sentence $\Phi_{\mathcal G}$ is an equality-free universal Horn sentence. Its
largest relation has arity $n+6$.  It contains
$O(1+n(|H|+|V_0|+|V_1|)+|T|^2)$ rules, each of length polynomial in the game
instance, so the construction is polynomial.  By
\cref{lem:corridor-ap-equivalence}, winning games map to failure of the AP and
non-winning games map to SAP.  Since deterministic $\ComplexityClass{2ExpTime}$ is
closed under complement, \cref{lem:binary-corridor-hardness} gives the lower
bound.  %
\end{proof}

\section{Conclusion}\label{sec:conclusion}

We have proved \ComplexityClass{ExpTime}-hardness of the amalgamation
decision problem for universal Horn sentences of maximum
relation arity three, and \ComplexityClass{2ExpTime}-hardness when arities
are unbounded. Both reductions produce equality-free sentences and
distinguish classes with SAP from classes without AP. They therefore
establish the same lower bounds for deciding strong amalgamation.

The proofs turn branching computations into obstructions to
amalgamation. At arity three, Horn rules parse a tree between two poles,
while an automaton simulation on one side detects accepted trees. With
unbounded arity, tuples address exponentially wide rows and Horn rules
assemble finite winning strategies. The binary corridor-game normal
form supplies, from every non-winning legal row, exactly two labelled moves
with unique legal successors,
so a universal choice can be represented by a single Horn clause.
In both cases, completeness of the remaining clauses and the pole
invariant allow strong amalgams to be constructed directly by Horn
closure.

The main question left open is whether the ADP for universal
Horn sentences is decidable in arbitrary arity, or already at maximum
arity three. A decision procedure for either fragment would have to control
amalgamation diagrams in arbitrary input theories; the closure invariants
proved here apply to the specific theories produced by the reductions.

It is also natural to ask whether the arity-three lower bound persists
over a fixed relational signature.
The first construction uses only one ternary symbol, but its unary and
binary symbols depend on the tree automaton. Understanding which parts
of that vocabulary can be fixed may help isolate the sources of
complexity in Horn amalgamation.

\section*{AI assistance}
The author used OpenAI's ChatGPT Sol 6.1 and Anthropic's
Claude Opus 5.5 to help checking the proofs in the manuscript;
the tools were also used to improve wording.

\bibliographystyle{amsplain}
\bibliography{references}

\providecommand{\bysame}{\leavevmode\hbox to3em{\hrulefill}\thinspace}
\providecommand{\MR}{\relax\ifhmode\unskip\space\fi MR }
\providecommand{\MRhref}[2]{%
  \href{http://www.ams.org/mathscinet-getitem?mr=#1}{#2}
}
\providecommand{\href}[2]{#2}
\begin{thebibliography}{10}

\bibitem{ArtaleEtAlSafety}
Alessandro Artale, Luca Geatti, Nicola Gigante, Andrea Mazzullo, and Angelo Montanari, \emph{Complexity of safety and {coSafety} fragments of linear temporal logic}, Proceedings of the AAAI Conference on Artificial Intelligence \textbf{37} (2023), no.~5, 6236--6244, Extended version: \url{https://arxiv.org/abs/2211.14913v2}.

\bibitem{AxelssonHagueKreutzerLangeLatte}
Roland Axelsson, Matthew Hague, Stephan Kreutzer, Martin Lange, and Markus Latte, \emph{Extended computation tree logic}, Logic for Programming, Artificial Intelligence, and Reasoning ({LPAR} 2010), Lecture Notes in Computer Science, vol. 6397, Springer, 2010, Extended version: \url{https://arxiv.org/abs/1006.3709}, pp.~67--81.

\bibitem{BodirskyKnauerStarkeASNP}
Manuel Bodirsky, Simon Kn{\"a}uer, and Florian Starke, \emph{{ASNP}: A tame fragment of existential second-order logic}, Beyond the Horizon of Computability ({CiE} 2020), Lecture Notes in Computer Science, vol. 12098, Springer, 2020, Preprint: \url{https://arxiv.org/abs/2001.08190}, pp.~149--162.

\bibitem{schrottenloher2022universal}
Manuel Bodirsky, Jakub Rydval, and Andr{\'e} Schrottenloher, \emph{Universal {Horn} sentences and the joint embedding property}, Discrete Mathematics \& Theoretical Computer Science \textbf{23} (2022), no.~2, 1--15.

\bibitem{ChandraKozenStockmeyerAlternation}
Ashok~K. Chandra, Dexter~C. Kozen, and Larry~J. Stockmeyer, \emph{Alternation}, Journal of the ACM \textbf{28} (1981), no.~1, 114--133.

\bibitem{ChlebusDominoTiling}
Bogdan~S. Chlebus, \emph{Domino-tiling games}, Journal of Computer and System Sciences \textbf{32} (1986), no.~3, 374--392.

\bibitem{TATA}
Hubert Comon, Max Dauchet, R{\'e}mi Gilleron, Florent Jacquemard, Denis Lugiez, Sophie Tison, and Marc Tommasi, \emph{Tree automata techniques and applications}, 1999, Version of 14 October 1999. \url{https://www.eecs.harvard.edu/~shieber/Projects/Transducers/Papers/comon-tata.pdf}.

\bibitem{rydval2022CSP}
Jakub Rydval, \emph{Using model theory to find decidable and tractable description logics with concrete domains}, Ph.D. thesis, Technische Universit{\"a}t Dresden, 2022, \url{https://tud.qucosa.de/id/qucosa%3A79907}.

\bibitem{RydvalICALP2024}
\bysame, \emph{Homogeneity and homogenizability: Hard problems for the logic {SNP}}, 51st International Colloquium on Automata, Languages, and Programming ({ICALP} 2024), Leibniz International Proceedings in Informatics, vol. 297, Schloss Dagstuhl -- Leibniz-Zentrum f{\"u}r Informatik, 2024, pp.~150:1--150:20.

\bibitem{RydvalInsideOut}
\bysame, \emph{Finitely bounded homogeneity turned inside-out}, Journal of Mathematical Logic (2025), 2550017, Expanded version of the {ICALP} 2024 paper; result numbering refers to \url{https://arxiv.org/abs/2108.00452v13}.

\bibitem{SeidlTreeAutomata}
Helmut Seidl, \emph{Deciding equivalence of finite tree automata}, SIAM Journal on Computing \textbf{19} (1990), no.~3, 424--437.

\end{thebibliography}

\end{document}